\documentclass[runningheads,10pt]{llncs}
\usepackage{pgf}
\usepackage {epic}
\usepackage {eepic}
\usepackage{latexsym}
\usepackage{mathrsfs}
\usepackage{amssymb}
\usepackage{graphics}
\usepackage{amsmath}

\authorrunning{N.~Alon, F.~V.~Fomin,  G.~Gutin, M.~Krivelevich, and S.~Saurabh}

\usepackage{amssymb}
\usepackage{mathrsfs}

\newtheorem{defn}[theorem]{Definition}
\newtheorem{fact}[theorem]{Fact}
\newtheorem{rem}[theorem]{Remark}

\newcommand{\2}{\vspace{0.2 cm}}

\title{Spanning directed trees with many leaves\thanks{Preliminary extended abstracts of this paper have been presented at
FSTTCS 2007  \cite{alonLNCS4596} and  ICALP 2007
\cite{AlonFGKS07fsttcs}}}

\author{
Noga Alon\inst{1} \and Fedor V. Fomin\inst{2}  \and Gregory
Gutin\inst{3} \and Michael Krivelevich\inst{1} \and Saket
Saurabh\inst{2}}
 \institute{Department of Mathematics, Tel Aviv
University\\ Tel Aviv 69978,
Israel\\\email{\{nogaa,krivelev\}@post.tau.ac.il}\and Department of
Informatics, University of Bergen\\ POB 7803, 5020 Bergen,
Norway\\
 \email{\{fedor.fomin,saket\}@ii.uib.no}
\and Department of Computer Science\\
 Royal Holloway, University of London\\
Egham, Surrey TW20 0EX, UK\\
\email{gutin@cs.rhul.ac.uk}
}

\begin{document}
\date{}
\maketitle

\begin{abstract}

The {\sc Directed Maximum Leaf Out-Branching}  problem is  to find
an out-branching (i.e. a rooted oriented spanning tree) in a given
digraph with the maximum number of leaves. In this paper, we obtain
two combinatorial results on the number of leaves in out-branchings.
We show that

\smallskip
\begin{itemize}
\item  every strongly connected $n$-vertex digraph $D$
 with minimum in-degree at least 3 has an out-branching
with at least $(n/4)^{1/3}-1$ leaves;
 \item if a strongly
connected digraph $D$ does not contain an out-branching with $k$
leaves, then the pathwidth of its underlying graph UG($D$) is
$O(k\log k)$. Moreover, if the digraph is acyclic, the pathwidth is
at most $4k$.
\end{itemize}
The last result implies that  it can be decided in time
$2^{O(k\log^2 k)}\cdot n^{O(1)}$ whether a strongly connected
digraph on $n$ vertices has an out-branching with at least $k$
leaves. On acyclic digraphs the running time of our algorithm is
$2^{O(k\log k)}\cdot n^{O(1)}$.


\end{abstract}

\section{Introduction}\label{introsec}
In this paper, we initiate the combinatorial and algorithmic study
of a natural generalization of the well studied {\sc Maximum Leaf
Spanning Tree (MLST)} problem on connected undirected graphs
\cite{bonsmaLNCS2747,DingJS01,estivill,fominGK06,FellowsMRS00,GalbiatMM97,GriggsW92,LuR98,Solis-Oba98}.
Given a digraph $D$, a subdigraph $T$ of $D$ is an {\em out-tree} if
$T$ is an oriented tree with only one vertex $s$ of in-degree zero
(called {\em the root}). If $T$ is a spanning out-tree, i.e.
$V(T)=V(D)$, then $T$ is called an {\em out-branching} of $D$. The
vertices of $T$ of out-degree zero are called {\em leaves}. The {\sc
Directed Maximum Leaf Out-Branching} (DMLOB) problem  is to find an
out-branching in a given digraph with the maximum number of leaves.

It is well-known that MLST is NP-hard for undirected graphs
\cite{GareyJ79}, which means that DMLOB is NP-hard for symmetric
digraphs (i.e., digraphs in which the existence of an arc $xy$
implies the existence of the arc $yx$) and, thus, for strongly
connected digraphs. We can show that DMLOB is NP-hard for acyclic
digraphs as follows: Consider a bipartite graph $G$ with bipartition
$X,Y$ and a vertex $s\not\in V(G)$. To obtain an acyclic digraph $D$
from $G$ and $s$, orient the edges of $G$ from $X$ to $Y$ and add
all arcs $sx$, $x\in X$. Let $B$ be an out-branching in $D$. Then
the set of leaves of $B$ is $Y\cup X'$, where $X'\subset X$, and for
each $y\in Y$ there is a vertex $z\in Z=X\setminus X'$ such that
$zy\in A(D)$. Observe that $B$ has maximum number of leaves if and
only if $Z\subseteq X$ is of minimum size among all sets
$Z'\subseteq X$ such that $N_G(Z')=X.$ However, the problem of
finding $Z'$ of minimum size such that $N_G(Z')=X$ is equivalent to
the Set Cover problem ($\{N_G(y)|\ y\in Y\}$ is the family of sets
to cover), which is NP-hard.

The combinatorial study of spanning trees with maximum number of
leaves in undirected graphs has an extensive history. Linial
conjectured around 1987 that every connected graph on $n$ vertices
with minimum vertex degree $\delta$ has a spanning tree with at
least $n(\delta-2)/(\delta+1)+c_\delta$ leaves, where $c_\delta$
depends on $\delta$. This is indeed the case for all $\delta \leq
5$. Kleitman and West \cite{KleitmanW91} and Linial and Sturtevant
\cite{LinialS87} showed that every connected undirected graph $G$ on
$n$ vertices with minimum degree at least $3$ has a spanning tree
with at least $n/4 + 2$ leaves. Griggs and Wu \cite{GriggsW92}
proved that the maximum number of leaves in a spanning tree is at
least $ n/2+2$ when $\delta=5$ and at least $2n/5+8/5$ when
$\delta=4$. All these results are tight. The situation is less clear
for $\delta \geq 6$; the first author observed that Linial's
conjecture is false for all large values of $\delta$. Indeed, the
results in \cite{Al} imply that there are undirected graphs with $n$
vertices and minimum degree $\delta$ in which no tree has more than
$(1-(1+o(1))\frac{\ln {(\delta+1)}}{\delta+1})n$ leaves, where the
$o(1)$-term tends to zero an $\delta$ tends to infinity, and this is
essentially tight. See also \cite{AS}, pp. 4-5 and \cite{CWY} for
more information.

 In this paper we prove an analogue of the Kleitman-West
result for directed graphs: every strongly connected digraph $D$ of
order $n$ with minimum in-degree at least 3 has an out-branching
with at least $(n/4)^{1/3}-1$ leaves.  We do not know whether this
bound is tight, however we show that there are strongly connected
digraphs with minimum in-degree 3 in which every out-branching has
at most $O(\sqrt{n})$ leaves.

Unlike its undirected counterpart which has attracted a lot of
attention in all algorithmic paradigms like approximation
algorithms~\cite{GalbiatMM97,LuR98,Solis-Oba98}, parameterized
algorithms~\cite{bonsmaLNCS2747,estivill,FellowsMRS00}, exact
exponential time algorithms~\cite{fominGK06} and also combinatorial
studies~\cite{DingJS01,GriggsW92,KleitmanW91,LinialS87}, the {\sc
Directed Maximum Leaf Out-Branching} problem has largely been
neglected until recently. The only paper we are aware of  is the
very recent paper \cite{drescher} that describes an
$O(\sqrt{\mbox{\sc opt}})$-approximation algorithms for DMLOB.

Our second combinatorial result relates the number of leaves in a
DMLOB of a directed graph $D$ with the pathwidth of its underlying
graph UG($D$). (We postpone the definition of pathwidth till the
next section.) If an  undirected graph $G$ contains
 a star $K_{1,k}$ as a minor, then it is possible to construct a spanning tree with at
 least $k$ leaves from this minor. Otherwise, there is no
 $K_{1,k}$  minor in $G$, and it is possible to prove  that the
pathwidth of $G$ is  $O(k)$. (See, e.g.  \cite{Bodlaender89}.)
Actually, a much more general result due to Bienstock et al.
\cite{BienstockRST91}) is that  any undirected graph of pathwidth at
least $k$, contains all trees on $k$ vertices as a minor. We prove a
result that can be viewed as a generalization of known bounds on the
number of leaves in a spanning tree of an undirected graph in terms
of its pathwidth, to strongly connected digraphs. We show that
either a strongly connected digraph $D$ has a DMLOB with at least
$k$ leaves or the pathwidth of UG($D$) is $O(k\log k)$. For an
acyclic digraph with a DMLOB having $k$ leaves, we prove  that the
pathwidth is at most $4k$. This almost matches the bound for
undirected graphs. These combinatorial results are useful in the
design of parameterized algorithms.

 In parameterized algorithms, for decision problems with input
size $n$, and a parameter $k$, the goal is to design an algorithm
with runtime $f(k)n^{O(1)}$, where $f$ is a function of $k$ alone.
(For DMLOB such a parameter is the number of leaves in the
out-tree.) Problems having such an algorithm are said to be fixed
parameter tractable (FPT). The book by Downey and Fellows
\cite{downey1999} provides an introduction to the topic of
parameterized complexity. For recent developments see the books by
Flum and Grohe \cite{FlumGrohebook} and by Niedermeier
\cite{Niedermeierbook06}.

The parameterized version of DMLOB is defined as follows: Given a
digraph $D$ and a positive integral parameter $k$, does $D$ contain
an out-branching with at least $k$ leaves? We denote the
parameterized versions of DMLOB by $k$-DMLOB. If in the above
definition we do not insist on an out-branching and ask whether
there exists an out-tree with at least $k$ leaves, we get the
parameterized {\sc Directed Maximum Leaf Out-Tree} problem (denoted
$k$-DMLOT).

Our combinatorial bounds, combined with dynamic programming on
graphs of bounded pathwidth imply the first parameterized algorithms
for $k$-DMLOB on strongly connected digraphs and acyclic digraphs.
We remark that the algorithmic results presented here also hold for
all digraphs if we consider $k$-DMLOT rather than $k$-DMLOB. This
answers an open question of Mike Fellows
\cite{Cesati06,fellows,gutin}. However, we mainly restrict ourselves
to $k$-DMLOB for clarity and the harder challenges it poses, and we
briefly consider $k$-DMLOT only in the last section.

Very recently, using a modification of our approach, Bonsma and Dorn
\cite{BoDo} proved that either an arbitrary digraph $D$ has an
out-branching with at most $k$ leaves or the pathwidth of UG($D'$)
is $O(k^3)$, where $D'$ is the digraph obtained from $D$ by deleting
all arcs not contained in any out-branching of $D.$ The bound
$O(k^3)$ is much larger than our bounds for strongly connected and
acyclic digraphs, but it suffices to allow Bonsma and Dorn to show
that $k$-DMLOB is FPT,  settling another open question of Fellows
\cite{fellows,gutin}.

\smallskip
This paper is organized as follows. In Section \ref{prelim} we
provide additional terminology and notation as well as some
well-known results. We introduce locally optimal out-branchings in
Section \ref{obsec}. Bounds on the number of leaves in maximum leaf
out-branchings of strongly connected and acyclic digraphs are
obtained in Section \ref{cbsec}. In Section
\ref{sec:pathewidth_of_underlying_graphs} we prove upper bounds on
the pathwidth of the underlying graph of strongly connected and
acyclic digraphs that do not contain out-branchings with at least
$k$ leaves.  In Section~\ref{dopsec} we conclude with  discussions
and open problems.

\section{Preliminaries}
\label{prelim} Let $D$ be a digraph. By $V(D)$ and $A(D)$ we
represent the vertex set and arc set of $D$, respectively. An {\em
oriented graph} is a digraph with no directed 2-cycle. Given a
subset $V'\subseteq V(D)$ of a digraph $D$, let $D[V']$ denote the
digraph induced by $V'$. The {\em underlying graph} UG($D$) of $D$
is obtained from $D$ by omitting all orientations of arcs and by
deleting one edge from each resulting pair of parallel edges. The
{\em connectivity components} of $D$ are the subdigraphs of $D$
induced by the vertices of components of UG($D$). A digraph $D$ is
{\em strongly connected} if, for every pair $x,y$ of vertices there
are directed paths from $x$ to $y$ and from $y$ to $x.$ A maximal
strongly connected subdigraph of $D$ is called a {\em strong
component}. A vertex $u$ of $D$ is an {\em in-neighbor} ({\em
out-neighbor}) of a vertex $v$ if $uv\in A(D)$ ($vu\in A(D)$,
respectively). The {\em in-degree} $d^-(v)$ ({\em out-degree}
$d^+(v)$) of a vertex $v$ is the number of its in-neighbors
(out-neighbors).

We denote by $\ell(D)$ the maximum number of leaves in an out-tree
of a digraph $D$ and by $\ell_s(D)$ we denote the maximum possible
number of leaves in an out-branching of a digraph $D$. When $D$ has
no out-branching, we write $\ell_s(D)=0$. The following simple
result gives necessary and sufficient conditions for a digraph to
have an out-branching. This assertion allows us to check whether
$\ell_s(D)>0$ in time $O(|V(D)|+|A(D)|)$.

\begin{proposition}[\cite{bang2000}]\label{iffoutb}
A digraph $D$ has an out-branching if and only if $D$ has a unique
strong component with no incoming arcs.
\end{proposition}

Let $P=u_1u_2\ldots u_q$ be a directed path in a digraph $D$. An arc
$u_iu_j$ of $D$ is a {\em forward} ({\em backward}) {\em arc for}
$P$ if $i\le j-2$ ($j<i$, respectively). Every backward arc of the
type $v_{i+1}v_i$ is called {\em double}.

For a natural number $n$, $[n]$ denotes the set $\{1,2,\ldots ,n\}.$

A {\em tree decomposition} of an (undirected) graph $G$ is a pair
$(X,U)$ where $U$ is a tree whose vertices we will call {\em nodes}
and $X=(\{X_{i} \mid i\in V(U)\})$ is a collection of subsets of
$V(G)$ such that
\begin{enumerate}
\item $\bigcup_{i \in V(U)} X_{i} = V(G)$,

\item for each edge $\{v,w\} \in E(G)$, there is an $i\in V(U)$
such that $v,w\in X_{i}$, and

\item for each $v\in V(G)$ the set of nodes $\{ i \mid v \in X_{i}
\}$ forms a subtree of $U$.
\end{enumerate}
The {\em width} of a tree decomposition $(\{ X_{i} \mid i \in V(U)
\}, U)$ equals $\max_{i \in V(U)} \{|X_{i}| - 1\}$. The {\em
treewidth} of a graph $G$ is the minimum width over all tree
decompositions of $G$.

If in the definitions of a tree decomposition and treewidth we
restrict $U$ to be a path, then we have the definitions of path
decomposition and pathwidth. We use the notation $tw(G)$ and $pw(G)$
to denote the treewidth and the pathwidth of a graph $G$.

We also need an equivalent definition of pathwidth in terms of
vertex separators with respect to a linear ordering of the vertices.
Let $G$ be a graph and let $\sigma=(v_1,v_2,\ldots ,v_n)$ be an
ordering of $V(G)$. For $j\in [n]$ put $V_j =\{v_i:\ i\in [j]\}$ and
denote by $\partial V_j$ all vertices of $V_j$ that have neighbors
in $V\setminus V_j.$ Setting $ vs(G,\sigma) = \max_{i\in [n]}
|\partial V_i | , $ we define the \emph{vertex separation}  of  $G$
as
\[
vs(G) = \min \{vs(G,\sigma) \colon \sigma \mbox{ is an ordering of }
V(G)\}.
\]

The following assertion is well-known.  It follows directly from the
results of Kirousis and Papadimitriou \cite{KirousisP85} on interval
width  of a graph, see also \cite{Kinnersley92}.

\begin{proposition}[\cite{Kinnersley92,KirousisP85}]\label{sovp_pw_vs}
For any graph $G$, $vs(G)=pw(G)$.
\end{proposition}



\section{Locally Optimal Out-Branchings}\label{obsec}
Our bounds are based on finding locally optimal out-branchings.
Given a digraph, $D$ and an out-branching $T$, we call a vertex {\em
leaf}, {\em link} and {\em branch} if its out-degree in $T$ is $0$,
$1$ and $\geq 2$ respectively. Let $S^{+}_{\geq 2 }(T)$ be the set
of branch vertices, $S^{+}_{1 }(T)$ the set of link vertices and
$L(T)$ the set of leaves in the tree $T$. Let $\mathscr{P}_2(T)$ be
the set of maximal paths consisting of link vertices. By $p(v)$ we
denote the {\em parent} of a vertex $v$ in $T$; $p(v)$ is the unique
in-neighbor of $v.$ We call a pair of vertices $u$ and $v$ {\em
siblings} if they do not belong to the same path from the root $r$
in $T$. We start with the following well known and easy to observe
facts.
\begin{fact}
 $|S^{+}_{\geq 2 }(T)| \leq |L(T)|-1$.
\end{fact}

\begin{fact}
 $|\mathscr{P}_2(T)| \leq 2 |L(T)|-1$.
\end{fact}

Now we define the notion of local exchange which is intensively used
in our proofs.
\begin{defn}
{\sc $\ell$-Arc Exchange ($\ell$-AE) optimal out-branching:} An
out-branching $T$ of  a directed graph $D$ with $k$ leaves is
$\ell$-AE optimal if
 for all arc subsets $F\subseteq A(T)$ and $X \subseteq A(D)-A(T)$ of size $\ell$,
 $(A(T)\setminus F) \cup X$  is either not an out-branching,  or an out-branching
with at most $ k$ leaves.
 In other words, $T$ is $\ell$-AE optimal if it can't be turned into an
out-branching with more leaves by exchanging
 $\ell$ arcs.
\end{defn}
Let us remark, that for every fixed $\ell$, an $\ell$-AE optimal
out-branching can be obtained in polynomial time. In our proofs we
use only $1$-AE optimal out-branchings. We need the following simple
properties of $1$-AE optimal out-branchings.
\begin{lemma}
\label{char1ae} Let $T$ be an $1$-AE optimal  out-branching rooted
at $r$ in a digraph $D$. Then the following holds:
\begin{itemize}
\item[(a)] For every pair of siblings $u,v\in V(T)\setminus L$ with  $d^{+}_T(p(v))
= 1$, there is no arc $e=(u,v)\in A(D)\setminus A(T)$;
\item[(b)] For every pair of vertices $u,v \notin L$, $d^{+}_T(p(v)) = 1$, which are
 on the same path from the root with $dist(r,u) < dist(r,v)$
there is no arc $e=(u,v)\in A(D)\setminus A(T)$ (here $dist(r,u)$ is
the distance to $u$ in $T$ from the root $r$);
\item[(c)] There is no arc $(v,r)$, $v \notin L$ such that the
directed cycle formed by the $(r,v)$-path and the arc $(v,r)$
contains a  vertex $x$ such that $d^{+}_T(p(x)) = 1$.
\end{itemize}
\end{lemma}
\begin{proof}
The proof easily follows from the fact that the existence of any of
these arcs contradicts the local optimality of $T$ with respect to
$1$-AE. \qed
\end{proof}

\section{Combinatorial Bounds}\label{cbsec}

We start with a lemma that allows us to obtain lower bounds on
$\ell_s(D)$.
\begin{lemma}\label{comblemma}
Let $D$ be a oriented graph of order $n$ in which every vertex is of
in-degree 2 and let $D$ have an out-branching. If $D$ has no
out-tree with $k$ leaves, then $n\le 4k^3.$
\end{lemma}
\begin{proof} Let us assume that $D$ has no out-tree with $k$ leaves. Consider
an out-branching $T$ of $D$ with $p<k$ leaves  which is $1$-AE
optimal. Let $r$ be the root of $T$.

We will bound the number $n$ of vertices in $T$ as follows. Every
vertex of $T$ is either a leaf, or a branch vertex, or a link
vertex. By Facts~1 and 2 we already have bounds on the number of
leaf and branch vertices as well as the number of maximal paths
consisting of link vertices. So to get an upper bound on $n$ in
terms of $k$, it suffices to bound the length of each maximal path
consisting of link vertices. Let us consider such a  path $P$ and
let $x,y$ be the first and last vertices of $P$, respectively.

The vertices of $V(T)\setminus V(P)$ can be partitioned into four
classes as follows:
\begin{itemize}\item[$(a)$] {\sf ancestor vertices}: the vertices which appear
before $x$ on the $(r,x)$-path of $T$;
\item[$(b)$]
{\sf descendant vertices }: the vertices appearing after the
vertices of $P$ on paths of $T$ starting at $r$ and passing through
$y$;
\item[$(c)$]{\sf sink vertices}: the vertices which are leaves but not descendant
vertices;
\item[$(d)$]
{\sf special vertices}: none-of-the-above vertices.
\end{itemize}

Let $P'=P-x$, let $z$ be the out-neighbor of $y$ on $T$ and let
$T_z$ be the subtree of $T$ rooted at $z$. By Lemma~\ref{char1ae},
there are no arcs from special or ancestor vertices to the path
$P'$. Let $uv$ be an arc of $A(D)\setminus A(P')$ such that $v\in
V(P').$ There are two possibilities for $u$: (i) $u\not\in V(P')$,
(ii) $u\in V(P')$ and $uv$ is backward for $P'$ (there are no
forward arcs for $P'$ since $T$ is 1-AE optimal). Note that every
vertex of type (i) is  either a descendant vertex or a sink. Observe
also that the backward arcs for $P'$ form a vertex-disjoint
collection of out-trees with roots at vertices that are not terminal
vertices of backward arcs for $P'$. These roots are terminal
vertices of arcs in which first vertices are descendant vertices or
sinks.

We denote by $\{u_1,u_2,\ldots, u_s\}$ and $\{v_1,v_2,\ldots, v_t\}$
the sets of vertices on $P'$ which have in-neighbors that are
descendant vertices and sinks, respectively. Let the out-tree formed
by backward arcs for $P'$ rooted at $w\in
\{u_1,\ldots,u_s,v_1,\ldots,v_t\}$ be denoted by $T(w)$ and let
$l(w)$ denote the number of leaves in $T(w).$ Observe that the
following is an out-tree rooted at $z$:
$$T_z \cup \{(in(u_1),u_1), \ldots, (in(u_s),u_s)\} \cup \bigcup_{i=1}^{s}T(u_i),$$
where $\{in(u_1),\ldots,in(u_s)\}$ are the in-neighbors of
$\{u_1,\ldots,u_s\}$ on $T_z.$ This out-tree has at least
$\sum_{i=1}^{s}l(u_i)$ leaves and, thus, $\sum_{i=1}^{s}l(u_i)\le
k-1.$ Let us denote the subtree of $T$ rooted at $x$ by $T_x$ and
let $\{in(v_1),\ldots,in(v_t)\}$ be the in-neighbors of
$\{v_1,\ldots,v_t\}$ on $T-V(T_x)$. Then we have the following
out-tree:
$$(T-V(T_x)) \cup \{(in(v_1),v_1), \ldots, (in(v_t),v_t)\} \cup
\bigcup_{i=1}^{t}T(v_i)$$ with at least $\sum_{i=1}^{t}l(v_i)$
leaves. Thus, $\sum_{i=1}^{t}l(v_i)\le k-1.$

Consider a path $R=v_0v_1\ldots v_r$ formed by backward arcs.
Observe that the arcs $\{v_iv_{i+1}:\ 0\le i\le r-1\}\cup
\{v_jv^+_j:\ 1\le j\le r\}$ form an out-tree with $r$ leaves, where
$v^+_j$ is the out-neighbor of $v_j$ on $P.$ Thus, there is no path
of backward arcs of length more than $k-1$. Every out-tree $T(w)$,
$w\in \{u_1,\ldots,u_s \}$ has $l(w)$ leaves and, thus, its arcs can
be decomposed into $l(w)$ paths, each of length at most $k-1$. Now
we can bound the number of arcs in all the trees $T(w)$, $w\in
\{u_1,\ldots,u_s \}$, as follows: $ \sum_{i=1}^{s}l(u_i)(k-1) \leq
(k-1)^2.$ We can similarly bound the number of arcs in all the trees
$T(w)$, $w\in \{v_1,\ldots,v_s \}$ by $(k-1)^2$. Recall that the
vertices of $P'$ can be either terminal vertices of backward arcs
for $P'$ or vertices in $\{u_1,\ldots,u_s,v_1,\ldots,v_t\}$. Observe
that $s+t\le 2(k-1)$ since $\sum_{i=1}^{s}l(u_i)\le k-1$ and
$\sum_{i=1}^{t}l(v_i)\le k-1.$

Thus, the number of vertices in $P$ is bounded from above by
$1+2(k-1)+2(k-1)^2$. Therefore,
\begin{eqnarray*}
n &=& |L(T)| + |S^{+}_{\geq 2 }(T)|+ |S^{+}_{ 1}(T)|\\
  &=& |L(T)| + |S^{+}_{\geq 2 }(T)|+ \sum_{P \in \mathscr{P}_2(T)} |V(P)|\\
  & \leq & (k-1) + (k-2) + (2k-3) (2k^2-2k+1)\\
  & < & 4k^3.
\end{eqnarray*}
Thus, we conclude that $n\le 4k^3.$ \qed\end{proof}

\begin{theorem}\label{main1}
Let $D$ be a strongly connected digraph with $n$ vertices.
\begin{enumerate}
\item[(a)] If  $D$ is an oriented graph with minimum in-degree at
least 2, then $\ell_s(D)\ge (n/4)^{1/3}-1.$ \item[(b)] If $D$ is a
digraph with minimum in-degree at least 3, then $\ell_s(D)\ge
(n/4)^{1/3}-1.$
\end{enumerate}
\end{theorem}
\begin{proof} Since $D$ is strongly connected, we have $\ell(D)=\ell_s(D)>0.$
Let $T$ be an 1-AE optimal out-branching of $D$ with maximum number
of leaves. (a) Delete some arcs from $A(D)\setminus A(T)$, if
needed, such that the in-degree of each vertex of $D$ becomes 2. Now
the inequality $\ell_s(D)\ge (n/4)^{1/3}-1$ follows from Lemma
\ref{comblemma} and the fact that $\ell(D)=\ell_s(D)$.

(b)  Let $P$ be the path formed in the proof of Lemma
\ref{comblemma}. (Note that $A(P)\subseteq A(T)$.) Delete every
double arc of $P$, in case there are any, and delete some more arcs
from $A(D)\setminus A(T)$, if needed, to ensure that the in-degree
of each vertex of $D$ becomes 2. It is not difficult to see that the
proof of Lemma \ref{comblemma} remains valid for the new digraph
$D$. Now the inequality $\ell_s(D)\ge (n/4)^{1/3}-1$ follows from
Lemma \ref{comblemma} and the fact that $\ell(D)=\ell_s(D)$. \qed
 \end{proof}

\begin{rem}\label{mainremark}
It is easy to see that Theorem \ref{main1} holds also for acyclic
digraphs $D$ with $\ell_s(D)>0$.
\end{rem}

While we do not know whether the bounds of Theorem \ref{main1} are
tight, we can show that no linear bounds are possible. The following
result is formulated for Part (b) of Theorem \ref{main1}, but a
similar result holds for Part (a) as well.

\begin{theorem}\label{exampleth}
For each $t\ge 6$ there is a strongly connected digraph $H_t$ of
order $n=t^2+1$ with minimum in-degree 3 such that
$0<\ell_s(H_t)=O(t).$
\end{theorem}
\begin{proof} Let $V(H_t)=\{r\}\cup \{u^i_1,u^i_2,\ldots ,u^i_{t}~|\ i\in
[t]\}$ and
\begin{eqnarray*}
A(H_t)& = & \left\{u^i_ju^i_{j+1},u^i_{j+1}u^i_j~|~\ i\in [t], j\in
\{0,1,\ldots ,t-3\}\right\} \\
& & \bigcup \left\{u^i_ju^i_{j-2}~|~ i\in [t], j\in
\{3,4,\ldots ,t-2\} \right \} \\
& & \bigcup \left\{u^i_ju^i_q ~|~\ i\in [t], t-3\le j\neq q\le
t\right\},
\end{eqnarray*}
where $u^i_0=r$ for every $i\in [t].$ It is easy to check that
$0<\ell_s(H_t)=O(t).$\qed
\end{proof}

%

\section{Pathwidth of underlying graphs and parameterized
algorithms}\label{sec:pathewidth_of_underlying_graphs}

By Proposition \ref{iffoutb}, an acyclic digraph $D$ has an
out-branching if and only if $D$ possesses a single vertex of
in-degree zero.

\begin{theorem}\label{dagth} Let $D$ be an acyclic
digraph with a single vertex of in-degree zero. Then either
$\ell_s(D)\ge k$ or the underlying undirected graph of $D$ is of
pathwidth at most $4k$ and we can obtain this path decomposition in
polynomial time.
\end{theorem}
\begin{proof}
Assume that $\ell_s(D)\le k-1$. Consider a $1$-AE optimal
out-branching $T$ of $D$. Notice that $|L(T)|\le k-1.$ Now remove
all the leaves and branch vertices from the tree $T$. The remaining
vertices form maximal directed paths consisting of link vertices.
Delete the first vertices of all paths. As a result we obtain a
collection $\cal Q$ of directed paths. Let $H=\cup_{P\in {\cal
Q}}P$. We will show that every arc $uv$ with $u,v\in V(H)$ is in
$H.$

Let $P'\in \cal Q$. As in the proof of Lemma \ref{comblemma}, we see
that there are no forward arcs for $P'$. Since $D$ is acyclic, there
are no backward arcs for $P'.$ Suppose $uv$ is an arc of $D$ such
that $u\in R'$ and $v\in P'$, where $R'$ and $P'$ are distinct paths
from $\cal Q$. As in the proof of Lemma \ref{comblemma}, we see that
$u$ is either a sink or a descendent vertex for $P'$ in $T$. Since
$R'$ contains no sinks of $T$, $u$ is a descendent vertex, which is
impossible as $D$ is acyclic. Thus, we have proved that $pw({\rm
UG}(H))=1.$

Consider a path decomposition of $H$ of width 1. We can obtain a
path decomposition of ${\rm UG}(D)$ by adding all the vertices of
$L(T)\cup S^{+}_{\geq 2 }(T)\cup F(T)$, where $F(T)$ is the set of
first vertices of maximal directed paths consisting of link vertices
of $T$, to each of the bags of a path decomposition of $H$ of width
1. Observe that the pathwidth of this decomposition is  bounded from
above by
$$|L(T)| + |S^{+}_{\geq 2 }(T)| + |F(T)|+1
\leq (k-1) + (k-2) + (2k-3)+1 \leq 4k -5.$$ The bounds on the
various sets in the inequality above follows from Facts $1$ and $2$.
This proves the theorem.
 \qed\end{proof}

\begin{corollary}\label{cor:acyclic}
For acyclic digraphs, the problem $k$-DMLOB can solved in time
$2^{O(k\log k)}\cdot n^{O(1)}$.
\end{corollary}
\begin{proof}
The proof of Theorem~\ref{dagth} can be easily turned into a
polynomial time algorithm to either build an out-branching of $D$
with at least $k$ leaves or to show that $pw({\rm UG}(D))\le 4k$ and
provide the corresponding path decomposition. A standard  dynamic
programming over the path (tree) decomposition (see e.g.
\cite{ArnborgP89-Li}) gives us an algorithm of running time
$2^{O(k\log k)}\cdot n^{O(1)}$.\qed
 \end{proof}

The following simple lemma is well known, see, e.g.,
\cite{chung1990}.

\begin{lemma}
\label{treesep} Let $T=(V,E)$ be an undirected tree and let $w ~:~V
\rightarrow \mathbb{R}^{+} \cup \{0\}$ be a weight function on its
vertices. There exists a vertex $v\in T$ such that the weight of
every subtree $T'$ of $T-v$ is at most $w(T)/2$, where $w(T)=\sum_{v
\in V} w(v)$.
\end{lemma}

Let $D$ be a strongly connected digraph with $\ell_s(D)=\lambda$ and
let $T$ be an out-branching of $D$ with $\lambda$ leaves. Consider
the following decomposition of $T$ (called a $\beta$-{\em
decomposition}) which will be useful in the proof of
Theorem~\ref{mainth}.

Assign weight 1 to all leaves of $T$ and weight 0 to all non-leaves
of $T$. By Lemma \ref{treesep}, $T$ has a vertex $v$ such that each
component of $T-v$ has at most $\lambda/2+1$ leaves (if $v$ is not
the root and its in-neighbor $v^-$ in $T$ is a link vertex, then
$v^-$ becomes a new leaf). Let $T_1,T_2,\ldots , T_s$ be the
components of $T-v$ and let $l_1,l_2,\ldots ,l_s$ be the numbers of
leaves in the components. Notice that $\lambda\le \sum_{i=1}^sl_i\le
\lambda+1$ (we may get a new leaf). We may assume that $l_s\le
l_{s-1}\le \cdots \le l_1\le \lambda/2+1.$ Let $j$ be the first
index such that $\sum_{i=1}^{j}l_i \geq \frac{\lambda}{2}+1.$
Consider two cases: (a) $l_{j}\le (\lambda+2)/4$ and (b) $l_{j}>
(\lambda+2)/4$. In Case (a), we have
$$ \frac{\lambda+2}{2}\leq \sum_{i=1}^{j}l_i \leq \frac{3(\lambda+2)}{4} \mbox{ and
}\frac{\lambda-6}{4} \le \sum_{i=j+1}^{s}l_i \le
\frac{\lambda}{2}.$$ In Case (b), we have $j=2$ and
$$ \frac{\lambda+2}{4}\le l_1\le  \frac{\lambda+2}{2} \mbox{ and }\frac{\lambda-2}{2}
\leq \sum_{i=2}^{s}l_i \le \frac{3\lambda+2}{4}.$$

Let $p=j$ in Case (a) and $p=1$ in Case (b). Add to $D$ and $T$ a
{\em copy} $v'$ of $v$ (with the same in- and out-neighbors). Then
the number of leaves in each of the out-trees
$$T'=T[\{v\}\cup (\cup_{i=1}^pV(T_i))] \mbox{ and } T''=T[\{v'\}\cup
(\cup_{i=p+1}^sV(T_i))]$$ is between $\lambda(1+o(1))/4$ and
$3\lambda(1+o(1))/4$. Observe that the vertices of $T'$ have at most
$\lambda+1$ out-neighbors in $T''$ and the vertices of $T''$ have at
most $\lambda+1$ out-neighbors in $T'$ (we add 1 to $\lambda$ due to
the fact that $v$ `belongs' to both $T'$ and $T''$).

Similarly to deriving $T'$ and $T''$ from $T$, we can obtain two
out-trees from $T'$ and two out-trees from $T''$ in which the
numbers of leaves are approximately between a quarter and three
quarters of the number of leaves in $T'$ and $T''$, respectively.
Observe that after $O(\log \lambda)$ `dividing' steps, we will end
up with $O(\lambda)$ out-trees with just one leaf, i.e., directed
paths. These paths contain $O(\lambda)$ copies of vertices of $D$
(such as $v'$ above). After deleting the copies, we obtain a
collection of $O(\lambda)$ disjoint directed paths covering $V(D)$.

\begin{theorem}\label{mainth} Let $D$ be a strongly connected digraph.
Then either $\ell_s(D)\ge k$ or the underlying undirected graph of
$D$ is of pathwidth $O(k \log k)$.
\end{theorem}
\begin{proof}
We may assume that $\ell_s(D)<k$. Let $T$ be be a 1-AE optimal
out-branching. Consider a $\beta$-decomposition of $T$. The
decomposition process can be viewed as a tree $\cal T$ rooted in a
node (associated with) $T$. The children of $T$ in $\cal T$ are
nodes (associated with) $T'$ and $T''$; the leaves of $\cal T$ are
the directed paths of the decomposition. The {\em first layer} of
$\cal T$ is the node $T$, the {\em second layer} are $T'$ and $T''$,
the {\em third layer} are the children of $T'$ and $T''$, etc. In
what follows, we do not distinguish between a node $Q$ of $\cal T$
and the tree associated with the node. Assume that $\cal T$ has $t$
layers. Notice that the last layer consists of (some) leaves of
$\cal T$ and that $t=O(\log k)$, which was proved above ($k\le
\lambda - 1$).

Let $Q$ be a node of $\cal T$ at layer $j$. We will prove that
\begin{equation}\label{ineq1}pw({\rm
UG}(D[V(Q)]))< 2(t-j+2.5)k\end{equation} Since $t=O(\log k)$,
(\ref{ineq1}) for $j=1$ implies that  the underlying undirected
graph of $D$ is of pathwidth $O(k \log k)$.

We first prove (\ref{ineq1}) for $j=t$ when $Q$ is a path from the
decomposition. Let $W=(L(T)\cup S^+_{\ge 2}(T)\cup F(T))\cap V(Q),$
where $F(T)$ is the set of first vertices of maximal paths of $T$
consisting of link vertices. As in the proof of Theorem \ref{dagth},
it follows from Facts 1 and 2 that $|W|<4k.$ Obtain a digraph $R$ by
deleting from $D[V(Q)]$ all arcs  in which at least one end-vertex
is in $W$ and which are not arcs of $Q$. As in the proof of Theorem
\ref{dagth}, it follows from Lemma \ref{char1ae} and 1-AE optimality
of $T$ that there are no forward arcs for $Q$ in $R$. Let $Q=v_1v_2
\dots v_q$. For every $j\in [q]$, let
 $V_j = \{v_i:\ i\in [j]\}$. If for some $j$ the set $V_j$
contained $k$ vertices, say $\{v_1',v_2',\cdots ,v_k'\}$, having
in-neighbors in the set $\{v_{j+1},v_{j+2}, \dots, v_q \}$, then $D$
would contain an out-tree with $k$ leaves formed by the path
$v_{j+1}v_{j+2} \dots v_q$ together with a backward arc terminating
at $v_i'$ from a vertex on the path for each $1\leq i \leq k$, a
contradiction. Thus $vs({\rm UG}(D_2[P]))\leq k.$ By
Proposition~\ref{sovp_pw_vs}, the pathwidth of ${\rm UG}(R)$ is at
most $k$. Let $(X_1, X_2, \ldots, X_s)$ be a path decomposition of
${\rm UG}(R)$ of width at most $k$. Then $(X_1\cup W, X_2\cup W,
\ldots, X_s\cup W)$ is a path decomposition of ${\rm UG}(D[V(Q)])$
of width less than $k+4k.$ Thus,
\begin{equation}\label{ineq2}pw({\rm
UG}(D[V(Q)]))<5k\end{equation}

Now assume that we have proved (\ref{ineq1}) for $j=i$ and show it
for $j=i-1$. Let $Q$ be a node of layer $i-1$. If $Q$ is a leaf of
$\cal T$, we are done by (\ref{ineq2}). So, we may assume that $Q$
has children $Q'$ and $Q''$ which are nodes of layer $i.$ In the
$\beta$-decomposition of $T$ given before this theorem, we saw that
the vertices of $T'$ have at most $\lambda+1$ out-neighbors in $T''$
and the vertices of $T''$ have at most $\lambda+1$ out-neighbors in
$T'$. Similarly, we can see that (in the $\beta$-decomposition of
this proof) the vertices of $Q'$ have at most $k$ out-neighbors in
$Q''$ and the vertices of $Q''$ have at most $k$ out-neighbors in
$Q'$ (since $k\le \lambda - 1$). Let $Y$ denote the set of the
above-mentioned out-neighbors on $Q'$ and $Q''$; $|Y|\le 2k.$ Delete
from $D[V(Q')\cup V(Q'')]$ all arcs in which at least one end-vertex
is in $Y$ and which do not belong to $Q'\cup Q''$

Let $G$ denote the obtained digraph. Observe that $G$ is
disconnected and $G[V(Q')]$ and $G[V(Q'')]$ are components of $G$.
Thus, $pw({\rm UG}(G))\le b$, where
\begin{equation}\label{ineq3}b=\max\{pw({\rm
UG}(G[V(Q')])),pw({\rm UG}(G[V(Q'')]))\}< 2(t-i+4.5)k\end{equation}
Let $(Z_1,Z_2,\ldots ,Z_r)$ be a path decomposition of $G$ of width
at most $b.$ Then $(Z_1\cup Y,Z_2\cup Y,\ldots ,Z_r\cup Y)$ is a
path decomposition of ${\rm UG}(D[V(Q')\cup V(Q'')]$) of width at
most $b+2k<2(t-i+2.5)k.$ \qed
 \end{proof}
Similar to the proof of Corollary~\ref{cor:acyclic}, we obtain the
following:
\begin{corollary}\label{cor:mainresult}
For a strongly connected digraph $D$, the problem $k$-DMLOB can be
solved in time $2^{O(k\log ^2k)}\cdot n^{O(1)}$.
\end{corollary}

\section{Discussion and Open Problems}\label{dopsec}

In this paper, we  initiated  the algorithmic and combinatorial
study of the {\sc Directed Maximum Leaf Out-Branching} problem.  In
particular, we showed that for every strongly connected digraph $D$
of order $n$ and with minimum in-degree at least 3,
$\ell_s(D)=\Omega(n^{1/3})$.  An interesting open combinatorial
question here is whether this bound is tight. If it is not, it would
be interesting to find the maximum number $r$ such that
$\ell_s(D)=\Omega(n^r)$ for every strongly connected digraph $D$ of
order $n$ and with minimum in-degree at least 3. It follows from our
results that $\frac{1}{3}\le r\le \frac{1}{2}.$

We also provided an algorithm of time complexity $2^{O(k\log^2
k)}\cdot n^{O(1)}$ which solves the $k$-DMLOB problem for a strongly
connected digraph $D$. The algorithm is based on a combinatorial
bound on the pathwidth of the underlying graph of $D$. Instead of
using results from
Section~\ref{sec:pathewidth_of_underlying_graphs}, one can use
Bodlaender's algorithm \cite{Bodlaender96} computing (for fixed $k$)
tree decomposition of width $k$ (if such a decomposition exists) in
linear time. Combined with our combinatorial bounds this yields a
linear time algorithm for
 $k$-DMLOB (for a strongly connected
digraphs). However, the exponential dependence of $k$ in
Bodlaender's algorithm is $c^{k^3}$ for some large constant $c$.

\medskip
Finally, let us observe that  while our results are for strongly
connected digraphs, they can be extended to a larger class of
digraphs. Notice that $\ell(D)\ge \ell_s(D)$ for each digraph $D$.
Let $\cal L$ be the family of digraphs $D$ for which either
$\ell_s(D)=0$ or $\ell_s(D)=\ell(D)$. The following assertion shows
that $\cal L$ includes a large number digraphs including all
strongly connected digraphs and acyclic digraphs (and, also, the
well-studied classes of semicomplete multipartite digraphs and
quasi-transitive digraphs, see \cite{bang2000} for the definitions).

\begin{proposition}[\cite{alonLNCS4596}]\label{L} Suppose that a digraph $D$ satisfies
the following property: for every pair $R$ and $Q$ of distinct
strong components of $D$, if there is an arc from $R$ to $Q$ then
each vertex of $Q$ has an in-neighbor in $R$. Then $D\in \cal L$.
\end{proposition}

Let $\cal B$ be the family of digraphs that contain out-branchings.
The results of this paper proved for strongly connected digraphs can
be extended to the class ${\cal L}\cap {\cal B}$ of digraphs since
in the proofs we use only the following property of strongly
connected digraphs $D$: $\ell_s(D)=\ell(D)>0$.

For a digraph $D$ and a vertex $v$, let $D_v$ denote the subdigraph
of $D$ induced by all vertices reachable from $v.$ Using the
$2^{O(k\log^2 k)}\cdot n^{O(1)}$ algorithm for $k$-DMLOB on digraphs
in ${\cal L}\cap {\cal B}$ and the facts that (i) $D_v\in {\cal
L}\cap {\cal B}$ for each digraph $D$ and vertex $v$ and (ii)
$\ell(D)=\max\{\ell_s(D_v)| v\in V(D)\}$ (for details, see
\cite{alonLNCS4596}), we can obtain an $2^{O(k\log^2 k)}\cdot
n^{O(1)}$ algorithm for $k$-DMLOT on {\em all} digraphs. For acyclic
digraphs, the running time can be reduced to $2^{O(k\log k)}\cdot
n^{O(1)}$.

\2

\noindent{\bf Acknowledgements.} Research of N. Alon and M.
Krivelevich was supported in part by USA-Israeli BSF grants and by
grants from the Israel Science Foundation. Research of F. Fomin was
supported in part by the Norwegian Research Council. Research of G.
Gutin was supported in part by EPSRC.


\end{document}